\documentclass[11pt,leqno,fleqn]{article}

\usepackage{tikz}
\usetikzlibrary{automata,arrows,shapes,snakes,positioning}
\tikzset{>=latex}
\tikzset{->-/.style={decoration={
  markings,
  mark=at position #1 with {\arrow{>}}},postaction={decorate}}}

\usepackage{latexsym}
\usepackage{amsmath}
\usepackage{amsthm}
\usepackage{amssymb}
\usepackage{amsfonts}
\usepackage{alltt} 

\usepackage[linesnumbered,ruled,vlined]{algorithm2e}
\let\oldnl\nl 
\newcommand{\nonl}{\renewcommand{\nl}{\let\nl\oldnl}}

\usepackage{amsbsy}
\usepackage[makeroom]{cancel}  

\usepackage{enumitem} 

\usepackage{nameref}
\usepackage[colorlinks=true,
            citecolor=blue,filecolor=blue,linkcolor=blue,urlcolor=blue,  
            linktoc=page,
            pagebackref=true]{hyperref} 
\usepackage{url}      
\usepackage{afterpage}  
\usepackage{float} 
\floatstyle{boxed} 

\usepackage{times}
\usepackage{bm}   
\usepackage{mathrsfs}  
\usepackage{graphicx}

\usepackage{fancybox}

\newcommand{\Hide}[1]{}

\definecolor{Myblue}{rgb}{0.8,0.85,1}         
\definecolor{Lightgreen}{rgb}{0.8,1,0.6}      
\newcommand{\HItxt}[1]{{\colorbox{Myblue}{#1}}}

\newif
\ifnote
\notetrue

\newif
\ifTR
\TRtrue

\newtheorem{theorem}{Theorem}
\newtheorem{lemma}[theorem]{Lemma}

\newtheorem{proposition}[theorem]{Proposition}

\newtheorem{restxxx}[theorem]{Restriction}

\newtheorem{agreexxx}[theorem]{Agreement}

\newtheorem{termxxx}[theorem]{Terminology}

\newtheorem{notxxx}[theorem]{Notation}

\newtheorem{assumxxx}[theorem]{Assumption}

\newtheorem{convenxxx}[theorem]{Convention}

\newtheorem{exaxxx}[theorem]{Example}

\newtheorem{exexxx}[theorem]{Exercise}

\newtheorem{remxxx}[theorem]{Remark}

\newtheorem{openxxx}[theorem]{Open Problem}
\newtheorem{conjxxx}[theorem]{Conjecture}

\newtheorem{defxxx}[theorem]{Definition}
\newenvironment{definition}[1]{\begin{defxxx}[\emph{#1}]\rm}%
   {\hfill\QED\end{defxxx}}
\newtheorem{defxxxsansQED}[theorem]{Definition}
   {\end{defxxxsansQED}}

\newenvironment{sketch}
{\smallskip\noindent\ignorespaces\textit{Proof Sketch.}}
{\hfill\QED\medskip}
\newtheorem{Prxxx}[theorem]{Proof}
{\end{Prxxx}} 

  {\addtolength{\leftskip}{#1}\addtolength{\rightskip}{#2}}{\par}

\usepackage[margin=0pt,justification=centerlast,font=small,format=hang,%
labelfont=bf,up,textfont=rm,up]{caption}

\newcommand{\Set}[1]{\{ #1 \}}
\newcommand{\SET}[1]{\bigl\{ #1 \bigr\}}

\newcommand{\B}{{\cal B}}

\newcommand{\N}{{\cal N}}

\newcommand{\bigO}[1]{{\cal O}\bigl(#1\bigr)} 
\newcommand{\bigOO}[1]{{\cal O}(#1)} 

\newcommand{\Let}[3]%
    {\textbf{\textsf{let}}\ {#1}\,{#2}\ \textbf{\textsf{in}}\;{#3}\,}
\newcommand{\Try}[3]%
    {\textbf{\textsf{try}}\ {#1} {#2}\ \textbf{\textsf{in}}\;{#3}\;}
\newcommand{\Mix}[3]%
    {\textbf{\textsf{mix}}\ {#1} {#2}\ \textbf{\textsf{in}}\;{#3}\;}
\newcommand{\LET}[3]%
    {\textbf{\textsf{let}}^{\bm{*}}\ {#1} {#2}\ \textbf{\textsf{in}}\;{#3}\;}
\newcommand{\Letrec}[3]%
    {\textbf{\textsf{letrec}}\ {#1} {#2}\ \textbf{\textsf{in}}\;{#3}\;}

\newcommand{\degreeSym}{\mathit{deg}} 
\newcommand{\degr}[2]{{\degreeSym}_{#1}(#2)}

\newcommand{\ie}{\textit{i.e.}}
\newcommand{\viz}{\textit{viz.}}
\newcommand{\eg}{\textit{e.g.}}
\newcommand{\QED}{{\Large $\square$}} 

\newcommand{\nreals}{\mathbb{R}_{+}}

\newcommand{\reals}{\mathbb{R}}

\newcommand{\intervals}[1]{{\cal I}(#1)}
\newcommand{\bfmath}[1]{{\text{\large $\boldsymbol{#1}$}}}

\newcommand{\size}[1]{|\,#1\,|}

\newcommand{\Power}[1]{\mathscr{P}\big(#1\big)}
\newcommand{\rest}[2]{\bfmath{[}#1\,|\,#2\bfmath{]}}

\newcommand{\head}[1]{\text{\em head}(#1)} 
\newcommand{\tail}[1]{\text{\em tail}(#1)} 

\newcommand{\OutF}[1]{\mathsf{OuterFace}(#1)}
\newcommand{\OutPlan}[2]{{{#1}\text{-}\mathsf{outerplanarity}}(#2)}

\newcommand{\CC}{\mathscr{C}}

\newcommand{\transA}[1]{{#1}^{\star}} 

\newcommand{\vertices}[1]{{\mathbf{V}(#1)}}
\newcommand{\edges}[1]{{\mathbf{E}(#1)}}

\newcommand{\edgesSharp}[1]{{\mathbf{E}_{\#}(#1)}}
\newcommand{\edgesIn}[1]{{\mathbf{E}_{\rm in}(#1)}}
\newcommand{\edgesOut}[1]{{\mathbf{E}_{\rm out}(#1)}}
\newcommand{\edgesIO}[1]{{\mathbf{E}_{\rm io}(#1)}}
\newcommand{\upperB}{\overline{\it c}}
\newcommand{\lowerB}{\underline{\it c}}

\newcommand{\maxFromToSym}{{\mathsf{maxFrom\!To}}}
\newcommand{\maxFromToFn}[1]{{\maxFromToSym}_{#1}}
\newcommand{\maxFromTo}[2]{{\maxFromToSym}_{#1}\big(#2\big)}
\newcommand{\maxFromToAftSym}{{\mathsf{maxFrom\!ToAft}}}
\newcommand{\maxFromToAftFn}[1]{{\maxFromToAftSym}_{#1}}
\newcommand{\maxFromToAft}[3]{{\maxFromToAftSym}_{#1}\big(#2;#3\big)}

\newcommand{\Complement}[1]{\overline{#1}}

\begin{document}


\setcounter{page}{1}     
\setcounter{tocdepth}{1} 
\ifTR
  \pagenumbering{roman} 
\else
\fi

\title{A Fixed-Parameter Linear-Time Algorithm \\
       to Compute Principal Typings of Planar Flow Networks} 
\author{Assaf Kfoury%
           \thanks{Partially supported by NSF awards CCF-0820138
           and CNS-1135722.} \\
        Boston University \\
        \ifTR Boston, Massachusetts \\ 
        \href{mailto:kfoury@bu.edu}{kfoury{@}bu.edu}
        \else \fi
        \Hide{
          \and   Benjamin Sisson%
           \footnotemark[1] \\
        Boston University   \\
        \ifTR Boston, Massachusetts \\ 
        \href{mailto:bmsisson@bu.edu}{bmsisson{@}bu.edu}
        \else \fi
        }
}

\ifTR
   \date{\today}
\else
   \date{} %
\fi
\maketitle
  \ifTR
     \thispagestyle{empty} 
  \else
  \fi

\ifTR
    \tableofcontents
    \newpage
\else
    \vspace{-.2in}
\fi

  \begin{abstract}
  \addcontentsline{toc}{section}{Abstract}

\noindent
We present an alternative and simpler method for computing principal typings
of flow networks. When limited to \emph{planar} flow networks, the method
can be made to run in fixed-parameter linear-time -- where the parameter
not to be exceeded is what is called the \emph{edge-outerplanarity} of
the networks' underlying graphs.


  \end{abstract}

  \newpage
  \pagenumbering{arabic}  

\section{Introduction}
\label{sect:intro}

\emph{Network typings} are algebraic or arithmetic formulations of interface
conditions that network components must satisfy in order to
interconnect with each other safely and correctly. A particular use of
network typings is to quantify desirable properties related to
resource management (\eg, percentage ranges of channel utilization,
mean delays between routers, etc., as well as flow conservation and
capacity constraints along channels), and to enforce them as invariant
properties across network interfaces. For a given network component
$\N$, a \emph{principal typing} for $\N$ is the most general -- or also
the most precise -- in the sense that it subsumes all other sound typings of $\N$.
More on this use of network typings is in several reports~\cite[and
the references therein]%
{BestKfoury:dsl11,Kfoury:sblp11,Kfoury:SCP2014}.
Computing efficiently \emph{principal typings} of networks is an
underlying concern in all these studies; new ways of computing them
more efficiently, under various conditions, continue to be
investigated.

In this report, we consider one version of network typings, here
simplified to account for only one quantity (\viz, flow) and under
only one restriction (\viz, flow must remain within pre-determined
upper bounds along all channels). A formal definition of network
typings that fits this simplified version is in
Section~\ref{sect:preliminaries}.  Our method for computing such
network typings efficiently (and more simply) is based on what is
called \emph{graph reassembling}.  When the underlying graph $G$ of a
network $\N$ is planar, our method runs in fixed-parameter linear
time, where the parameter to be bounded is called
the \emph{edge-outerplanarity} of $G$.
We next explain these two notions: \emph{graph reassembling} and
\emph{edge-outerplanarity}.

One way of understanding the \emph{reassembling} of a simple
undirected graph $G$ is this: It is the process of cutting every edge of
$G$ in two halves, and then splicing the two halves of every edge, one
by one in some order, in order to recover the original $G$. We thus
start from one-vertex components, with one component for each vertex $v$
and each with $\degr{}{v}$ dangling half edges,%
\footnote{$\degr{}{v}$ is the degree of vertex $v$, \ie, the number of
  edges incident to $v$, both entering $v$ and exiting $v$.
  }
and then gradually reassemble larger and larger components of the
original $G$ until $G$ is fully reassembled. One optimization
associated with graph reassembling is to keep the number of dangling
half edges of each reassembled component as small as possible. Graph
reassembling and associated optimization problems are examined in
earlier reports on network
analysis~\cite{Kfoury:SCP2014,SouleBestKfouryLapets:eoolt11,%
kfoury+mirzaei:2017,kfoury+mirzaei:2017B}. A formal definition of
\emph{graph reassembling} -- different from, but
equivalent to, the preceding informal definition -- 
is in Section~\ref{sect:our-result}.

As for the notion of \emph{edge-outerplanarity} of planar graphs, it
is distinct but closely related to the usual notion of outerplanarity,
and was introduced in earlier studies for other purposes (\eg,
disjoint paths in sparse graphs, as in~\cite{bentz2009}). As with
outerplanarity, for a fixed edge-outerplanarity $k$, the number $n$ of
vertices in a planar graph can be arbitrarily large. Our main result
can be re-phrased thus: Our main result can be re-phrased thus: For
the class ${\CC}_k$ of planar flow networks whose edge-outerplanarity
is bounded by a fixed $k\geqslant 1$, there is an algorithm which,
given an arbitrary $\N\in{\CC}_k$, computes a principal typing for
$\N$ in time $\bigOO{n}$ where $n = \size{\N}$.


\section{Preliminary Notions}
\label{sect:preliminaries}

We review several standard notions, add new notions specially adapted
to our needs in this paper, and fix our notational conventions.


\subsection*{Flow Networks:}

A flow network is a pair of the form $\N = (G,c)$
where $G$ is a directed graph without self-loops and without
multi-edges (in the same direction),%
   \footnote{\label{foot:two-edge-cycle}
   However, $G$ may contain \emph{two-edge cycles},
   \ie, two edges $e_1$ and $e_2$ such that $\head{e_1} = \tail{e_2}$
   and $\tail{e_1} = \head{e_2}$.
   }
and $c: \edges{G}\to\nreals$ is a function that assigns
an \emph{upper-bound capacity} to every edge $e$.
We write $\vertices{G}$ and $\edges{G}$ for the set of
vertices and the set of edges of $G$, respectively.

For reasons that become clear later, we do not identify subsets of
$\vertices{G}$ as `sources' and `sinks' of $\N$, following usual
conventions.  Instead, we allow some members of $\edges{G}$ to be
`dangling' edges. An edge $e \in\edges{G}$ is \emph{dangling} if it is
incident to only one vertex $v\in\vertices{G}$, for which there are
two cases, where we write `$\bot$' to mean `undefined':
\begin{itemize}[itemsep=1pt,parsep=2pt,topsep=2pt,partopsep=0pt]
\item $\head{e}=v$ and $\tail{e} = \bot$, in which case $e$ is
      an \emph{input edge}, or
\item $\tail{e}=v$ and $\head{e} = \bot$, in which case $e$ is
      an \emph{output edge}.
\end{itemize}
 $\edgesIn{G}$ denotes the set of input edges and $\edgesOut{G}$
 the set of output edges.
An edge $e\in\edges{G}$ is \emph{not dangling} if it is incident to
two distinct vertices $v,w\in\vertices{G}$ with $v = \tail{e}$ and
$w = \head{e}$.  The set of edges that are \emph{not dangling} is
denoted $\edgesSharp{G}$.
The three sets $\SET{\edgesIn{G}, \edgesOut{G}, \edgesSharp{G}}$
form a $3$-part partition of $\edges{G}$, \ie, they are 
pairwise disjoint and:
\[
  \edges{G} = \edgesIn{G} \cup \edgesOut{G} \cup \edgesSharp{G} .
\]
We write $\edgesIO{G}$ for the union $\edgesIn{G}\cup\edgesOut{G}$.

As usual, a \emph{flow} in the network $\N$ is a function
$f : \edges{G}\to\nreals$. If $X\subseteq \edges{G}$, we write
$f(X)$ for the summation $\sum\Set{\,f(e)\,|\,e\in X\,}$.
The flow $f$ is \emph{feasible} if it satisfies
the two standard conditions:
\begin{itemize}[itemsep=1pt,parsep=2pt,topsep=2pt,partopsep=0pt]
\item \emph{flow conservation} at every vertex $v\in\vertices{G}$,
       \ie, if $X$ and $Y$ are all the edges entering $v$ and
       exiting $v$, respectively, then $f(X) = f(Y)$,
\item \emph{capacity constraint} at every edge $e\in\edges{G}$,
      \ie, $ f(e) \leqslant c(e)$.
\end{itemize}
An \emph{input-output assignment} (or an \emph{IO assignment}) for the
network $\N = (G,c)$ is a function $g : \edgesIO{G}\to\nreals$.
The \emph{restriction} of a flow $f :\edges{G}\to\nreals$ to the subset
$\edgesIO{G}\subseteq\edges{G}$, denoted $\rest{f}{\edgesIO{G}}$, is an
IO assignment.
The \emph{value} of the flow $f$, denoted $\size{f}$, is
$f\big(\edgesIn{G}\big)$ or, equivalently, $f\big(\edgesOut{G}\big)$.

If $f, f' : \edges{G}\to\nreals$ are two flows in $\N$, then $f+f'$
denotes their sum: 
$(f+f')(e) \triangleq f(e) + f'(e)$ for every edge $e\in\edges{G}$.


\subsection*{Network Typings:}

Let $\Power{\edgesIO{G}}$ be the powerset of $\edgesIO{G}$
and $\intervals{\reals}$ the set of closed real intervals:
\[
  \Power{\edgesIO{G}} \triangleq \SET{\,X\;\big|\;X\subseteq \edgesIO{G}\,}
  \quad\text{and}\quad
  \intervals{\reals} \triangleq
  \SET{\,[r_1,r_2]\;\big|\;r_1,r_2\in\reals\text{ and } r_1\leqslant r_2 \,}.
\]  
A \emph{typing} for the network $\N = (G,c)$ is a map $\tau$ of the form:
\[
         \tau: \Power{\edgesIO{G}} \to\intervals{\reals}.
\]
 If $X \subseteq \edgesIO{G}$ with
 $A = X\cap \edgesIn{G}$ and $B = X\cap \edgesOut{G}$,
 we may write $\tau(A,B)$ instead of $\tau(X)$.

 An IO assginment $g: \edgesIO{G}\to \nreals$ \emph{satisfies}
 the typing $\tau: \Power{\edgesIO{G}} \to\intervals{\reals}$ iff
 for every $A \subseteq \edgesIn{G}$ and every $B \subseteq \edgesOut{G}$
 it holds that:
\[
  g(A) - g(B) \in \tau(A,B).
\]
We can view the difference $g(A) - g(B)$ as expressing the excess flow
that enters at $A$ but does not exit from $B$, which may be positive
or negative. Only when $A = \edgesIn{G}$ and $B = \edgesOut{G}$ do we
have $g(A) - g(B) = 0$.

A flow $f: \edges{G}\to \nreals$ \emph{satisfies} the typing $\tau$ if
its restriction $\rest{f}{\edgesIO{G}}$ satisfies $\tau$.

\begin{definition} 
{Principal Typings}
\label{def:principal}
A typing $\tau: \Power{\edgesIO{G}} \to\intervals{\reals}$ for the
network $\N = (G,c)$ is \emph{principal} iff two conditions are satisfied:
\begin{itemize}[itemsep=1pt,parsep=2pt,topsep=2pt,partopsep=0pt]
\item If a flow $f: \edges{G}\to \nreals$ is feasible, then
      $f$ satisfies $\tau$.
\item If an IO assignment $g: \edgesIO{G}\to \nreals$ satisfies $\tau$,
      then $g$ can be extended to a feasible flow, \\
      \ie, there is feasible flow $f$ such that $g = \rest{f}{\edgesIO{G}}$.
\end{itemize}
The first condition is the \emph{completeness} of $\tau$, the second
condition is the \emph{soundness} of $\tau$. A minimum requirement on
any typing $\tau$ for $\N$ is that it be sound; if $\tau$ is also complete
for $\N$, and therefore principal for $\N$, then $\tau$ is the `most precise'
formulation of the condition for connecting $\N$ with other networks.
\end{definition} 


\subsection*{Two Special Functions:}

Relative to a flow network $\N = (G,c)$, we define two functions written as:
\[
        \maxFromTo{\N}{A_1,B_1}
          \quad\text{and}\quad
        \maxFromToAft{\N}{A_1,B_1}{A_2,B_2},
\]
where $A_1, A_2\subseteq\edgesIn{G}$ and $B_1,B_2\subseteq\edgesOut{G}$.
Whenever `$\N$' is understood from the context,
we omit the subscript `$\N$' and write instead:
\[
        \maxFromTo{}{A_1,B_1}
          \quad\text{and}\quad
        \maxFromToAft{}{A_1,B_1}{A_2,B_2}.
\]
The meaning of the first function is given by~(\ref{eq:a})
or~(\ref{eq:b}):
  \begin{alignat}{5}
  \label{eq:a}
    & \maxFromTo{}{A,B}\triangleq
  \max\;\SET{\,\HItxt{$f(A)$}\;\big|\;\text{$f:\edges{G}\to\nreals$ is feasible
      and $f(\Complement{A}) = f(\Complement{B}) = 0$}\,},
    \\[1.01ex]
  \label{eq:b}
    & \maxFromTo{}{A,B}\triangleq
  \max\;\SET{\,\HItxt{$f(B)$}\;\big|\;\text{$f:\edges{G}\to\nreals$ is feasible
      and $f(\Complement{A}) = f(\Complement{B}) = 0$}\,}.
  \end{alignat}
(\ref{eq:a}) and~(\ref{eq:b}) are identical except for the highlighted parts.
It is an easy exercise (omitted) to show~(\ref{eq:a})
and~(\ref{eq:b}) are equivalent definitions.%
    \footnote{There are different ways of proving the equivalence
    of~(\ref{eq:a}) and~(\ref{eq:b}). One particular simple way
    is by induction on the number $m$ of edges for a fixed number
    $n$ of vertices. Another simple way is to remove all input
    edges in $\Complement{A}$ and all output edges in $\Complement{B}$,
    then join all input edges in $A$ to a fresh input edge $e_{\text{in}}$
    and all output edges in $B$ to a fresh output edge $e_{\text{out}}$,
    and then consider maximum flows from $e_{\text{in}}$ to
    $e_{\text{out}}$ in the thus-modified network.
    }
Informally, 
$\maxFromTo{}{A,B}$ is the value of a maximum flow
$f:\edges{G}\to\nreals$ from $A\subseteq\edgesIn{G}$ to
$B\subseteq\edgesOut{G}$ when flow is blocked from entering
$\Complement{A}$ and from exiting $\Complement{B}$.

\Hide{ 

The meaning of the first function is as follows, if there are both a lower
bound capacity and an upper bound capacity:

THE FOLLOWING WILL NOT WORK. IT SEEMS YOU HAVE TO GO FROM feasible flows TO
feasible path-assignments.

Check out: main.pdf in directory Generalized Max Flow.

\begin{itemize}
\item $\maxFromTo{}{A,B}$ is the value of a
  maximum flow $f:\edges{G}\to\nreals$
  from $A\subseteq\edgesIn{G}$ to $B\subseteq\edgesOut{G}$:
    \begin{alignat}{5}
    \label{eq:1}
    & \maxFromTo{}{A,B}\triangleq
    \\
    &
  \max\SET{\,f(A)-f(\Complement{B})
  \;\big|\;\text{$f:\edges{G}\to\nreals$ is feasible}\,}
    \nonumber
    \\[2.2ex]
    \label{eq:2}
    & \maxFromTo{}{A,B}\triangleq
    \\
    &
  \max\SET{\,f(B)-f(\Complement{A})
  \;\big|\;\text{$f:\edges{G}\to\nreals$ is feasible}\,}
    \nonumber
  \end{alignat}
\end{itemize}          
The two definition are equivalent because
$ f(A)+f(\Complement{A}) = f(B)+f(\Complement{B}) $ for all
feasible $f:\edges{G}\to\nreals$.
Informally, $\maxFromTo{}{A,B}$ returns the value of a
maximum flow from $A$ to $B$.
}

In the case of the second function $\maxFromToAft{}{A_1,B_1}{A_2,B_2}$,
it will always be the case that:
\begin{itemize}[itemsep=1pt,parsep=2pt,topsep=2pt,partopsep=0pt]
\item[] \emph{either} $A_1\cap A_2 =\varnothing$ and $B_1 = B_2$,
\item[] \emph{or}\quad\ \ $A_1 = A_2$ and $B_1\cap B_2 =\varnothing$.
\end{itemize}
For the first of these two cases,
the meaning of $\maxFromToAft{}{A_1,B}{A_2,B}$ is given by~(\ref{eq:c})
or~(\ref{eq:d}), where $A_1, A_2 \subseteq\edgesIn{G}$
and $A_1\cap A_2 =\varnothing$:
  \begin{alignat}{5}
  \label{eq:c}
    & \maxFromToAft{}{A_1,B}{A_2,B}\triangleq 
    \,\max \big\{\,\HItxt{$f(A_1)$}
    \;\big|\; && \text{$f+f':\edges{G}\to\nreals$ is feasible
      for some flow $f'$ }
    \\
    & &&
    \text{such that $f'(A_2) = f'(B) = \maxFromTo{}{A_2,B}$ } 
  \nonumber
    \\
    & &&
    \text{and $(f+f')(\Complement{A_1\cup A_2}) = (f+f')(\Complement{B}) = 0$}
    \,\big\},
  \nonumber
    \\[1.091ex]
  \label{eq:d}
    & \maxFromToAft{}{A_1,B}{A_2,B}\triangleq 
    \,\max \big\{\, \HItxt{$f(B)$}
    \;\big|\; && \text{$f+f':\edges{G}\to\nreals$ is feasible 
      for some flow $f'$ }  
    \\
    & &&
    \text{such that $f'(A_2) = f'(B) = \maxFromTo{}{A_2,B}$ } 
  \nonumber
    \\
    & &&
    \text{and $(f+f')(\Complement{A_1\cup A_2}) = (f+f')(\Complement{B}) = 0$}
    \,\big\}.
  \nonumber  
\end{alignat}
(\ref{eq:c}) and~(\ref{eq:d}) are identical except for the highlighted parts.
For the second case of the function $\maxFromToAftSym$,
the meaning of $\maxFromToAft{}{A,B_1}{A,B_2}$ is given by~(\ref{eq:e})
or~(\ref{eq:f}), where $B_1, B_2 \subseteq\edgesOut{G}$
and $B_1\cap B_2 =\varnothing$:
  \begin{alignat}{5}
  \label{eq:e}
    & \maxFromToAft{}{A,B_1}{A,B_2}\triangleq 
    \,\max \big\{\,\HItxt{$f(B_1)$}
    \;\big|\; && \text{$f+f':\edges{G}\to\nreals$ is feasible
      for some flow $f'$ }
    \\
    & &&
    \text{such that $f'(A) = f'(B_2) = \maxFromTo{}{A,B_2}$ } 
  \nonumber
    \\
    & &&
    \text{and $(f+f')(\Complement{A}) = (f+f')(\Complement{B_1\cup B_2}) = 0$}
    \,\big\},
  \nonumber
    \\[1.091ex]
  \label{eq:f}
    & \maxFromToAft{}{A,B_1}{A,B_2}\triangleq 
    \,\max \big\{\,\HItxt{$f(A)$}
    \;\big|\; && \text{$f+f':\edges{G}\to\nreals$ is feasible
      for some flow $f'$ }
    \\
    & &&
    \text{such that $f'(A) = f'(B_2) = \maxFromTo{}{A,B_2}$ } 
  \nonumber
    \\
    & &&
    \text{and $(f+f')(\Complement{A}) = (f+f')(\Complement{B_1\cup B_2}) = 0$}
    \,\big\},
  \nonumber
\end{alignat}
(\ref{eq:e}) and~(\ref{eq:f}) are identical except for the highlighted parts.
Just as~(\ref{eq:a}) and~(\ref{eq:b}) are equivalent, so 
too~(\ref{eq:c}) and~(\ref{eq:d}) are equivalent, and~(\ref{eq:e})
and~(\ref{eq:f}) are equivalent, and by the same reasoning.

Informally, the meaning of $\maxFromToAft{}{A_1,B}{A_2,B}$ and
$\maxFromToAft{}{A,B_1}{A,B_2}$ is as follows:
\begin{itemize}[itemsep=1pt,parsep=2pt,topsep=2pt,partopsep=0pt]
\item
      $\maxFromToAft{}{A_1,B}{A_2,B}$ returns the value of a maximum
      flow \emph{from $A_1$ to $B$, after} \\ a maximum flow has been
      already directed \emph{from $A_2$ to $B$},
\item
      $\maxFromToAft{}{A,B_1}{A,B_2}$ returns the value of a maximum
      flow \emph{from $A$ to $B_1$, after} \\ a maximum flow has been
      already directed \emph{from $A$ to $B_2$}.
\end{itemize}
The following lemma is used in the induction in Section~\ref{sect:reassembling}.

\begin{lemma}
\label{lem:relating-the-two-functions}
The functions $\maxFromToSym$ and
$\maxFromToAftSym$ are related by the following equalities:
\begin{enumerate}
\item[$(\dag)$]
   For all $A_1,A_2\subseteq\edgesIn{G}$ and $B\subseteq\edgesOut{G}$ such that
   $A_1\cap A_2 = \varnothing$:
   \[
   \maxFromToAft{}{A_1,B}{A_2,B}\ =
   \ \maxFromTo{}{A_1\cup A_2,B} - \maxFromTo{}{A_2,B}
   \]
\item[$(\ddag)$]
   For all $A\subseteq\edgesIn{G}$ and $B_1,B_2\subseteq\edgesOut{G}$
   such that $B_1\cap B_2 = \varnothing$:
   \[
   \maxFromToAft{}{A,B_1}{A,B_2}\ =
   \maxFromTo{}{A,B_1\cup B_2} - \maxFromTo{}{A,B_2}.
   \]
\end{enumerate}
\end{lemma}

\begin{sketch}
The proof of $(\dag)$ and $(\ddag)$ are essentially the same, and it suffices
to focus on $(\dag)$. Hence, from (\ref{eq:c}) we
need to show that $(\dag)$ is true; in fact, what is more, (\ref{eq:c})
and $(\dag)$ imply each other. This is easy to see by conservation
of flow through the network. A more formal proof is to prove the equivalence
of (\ref{eq:c}) and $(\dag)$ for every component of $\N$ as it is reassembled
inductively in Section~\ref{sect:reassembling} -- said differently still,
given the definition in (\ref{eq:c}), the equality $(\dag)$ is an invariant
of the induction -- starting with the one-vertex
components and finishing with the full network $\N$.
\end{sketch}

The next lemma is used in the proof of our main result,
Theorem~\ref{thm:our-result}.

\begin{lemma}
\label{lem:typing}
  Let $\tau: \Power{\edgesIO{G}} \to\intervals{\reals}$ be the
  principal typing of the flow network $\N = (G,c)$.
  For all $A\subseteq\edgesIn{G}$ and $B\subseteq\edgesOut{G}$, it holds 
  that: 
\[
   \tau(A,B) = [r_1,r_2] 
     \quad\text{iff}\quad
   r_1 = -\maxFromTo{}{\Complement{A},B}
   \text{\ and\ }   r_2 = \maxFromTo{}{A,\Complement{B}} .
\]
\end{lemma}

\begin{sketch}
Somewhat informally, using flow conservation through the network, this
is a straightforward consequence of the definitions of `network
typings' and the function $\maxFromToFn{\N}$. More formal, but less
transparent, is a proof by induction, as $\N$ is reassembled
inductively from the one-vertex components to the full network
$\N$, as in Section~\ref{sect:reassembling}. All formal details omitted.
\end{sketch}


\section{Reassembling the Network}
\label{sect:reassembling}

 Given a flow network $\N = (G,c)$, let
 $m = \size{\edgesSharp{G}}$ and $n = \size{\edges{G}}$. Note that $m$
 does not include a count of the edges in
 $\edgesIn{G}\cup\edgesOut{G}$. Starting from $n$ one-vertex
 components, which we denote:
 \[
   {\N}_1 = (G_1,c),\quad {\N}_2 = (G_2,c),\ \ \ldots\ \ ,
   \quad {\N}_n = (G_n,c),
 \]
 with one for each of the $n$ vertices,
 we splice the two halves of each of the $m$ edges in $\edgesSharp{G}$,
 one by one in some order, until the full network $\N = (G,c)$ is
 reassembled:
  \[
   {\N}_{n+1} = (G_{n+1},c),\quad {\N}_{n+2} = (G_{n+2},c),\ \ \ldots\ \ ,
   \quad {\N}_{n+m} = (G_{n+m},c),
 \]
 where  ${\N}_{n+m} = \N$. For every $i\geqslant n+1$, the graph $G_i$ is 
 directed and \emph{connected}, though not necessarily \emph{strongly connected},
 and has at least two vertices.

 For every $k = n+1,\ldots,n+m$, the new network component ${\N}_k$ is the
 result of splicing the two dangling halves of some non-dangling edge $e$ in
 the initial $G$. If the two halves of $e$ are $e_1$ and $e_2$, then the
 new ${\N}_k$ is related to the preceding network components
 $\Set{{\N}_1,\ldots, {\N}_{k-1}}$ in one of two ways:
 \begin{description}[itemsep=1pt,parsep=2pt,topsep=2pt,partopsep=0pt]
 \item[Case 1:]\label{case1}
   There are two distinct network components ${\N}_i$
   and ${\N}_j$ such that $i < j < k$, with $e_1$ an input (or output)
   edge in ${\N}_i$ and $e_2$ an output (or, resp., input) edge in
   ${\N}_j$.
 \item[Case 2:]\label{case2}
   There is one network component ${\N}_i$
   such that $i < k$, with both $e_1$ an input (or output) edge 
   and $e_2$ an output (or, resp., input) edge in ${\N}_i$.   
 \end{description}
 For every $i\in\Set{1,\ldots,n+m}$, define the
 quantities:
 \[
 p_i \triangleq \size{\edgesIn{G_i}}
 \quad\text{and}\quad
 q_i \triangleq \size{\edgesOut{G_i}}.
 \]
 Thus, $p_i+q_i$ is the total number of dangling edges (input edges
 and output edges) in $G_i$, what is also called the
 \emph{edge-boundary degree} of $G_i$ (the number of edges that
 connect vertices inside $G_i$ with vertices outside $G_i$).
 
 We do not worry now about the order in which the reassembling is
 carried out in this section. Later we specify an order with which we 
 obtain the result claimed in the report's title. Define:
 \[
 \delta\, \triangleq
 \ \max\, \Set{\,p_i+q_i\;|\; 1\leqslant i\leqslant n+m\,} .
 \]
 Thus, $\delta$ is the least upper bound on the edge-boundary degrees of
 $\Set{G_1,\ldots,G_{n+m}}$. Our next task is to determine the function
 $\maxFromToFn{{\N}_i}$ for every $i = 1,\ldots,n+m$. We do this by
 induction on $i$.

 \paragraph{Basis step.} Let $i \in \Set{1,\ldots,n}$. Each ${\N}_i = (G_i,c)$ is
 a one-vertex component. If $\vertices{G_i} = \Set{v}$, then
 $p_i+q_i = \degr{}{v}$. It is straightforward to compute
 $\maxFromTo{{\N}_i}{A,B}$ for every $A\subseteq\edgesIn{G_i}$
 and every $B\subseteq\edgesOut{G_i}$. All details omitted.

 \paragraph{Induction hypothesis (IH).}
 Let $k \in \Set{n+1,\ldots,n+m-1}$.  For every $i \leqslant k-1$,
 assume $\maxFromTo{{\N}_i}{A,B}$ has been already determined for
 every $A\subseteq\edgesIn{G_i}$ and every $B\subseteq\edgesOut{G_i}$.

 \paragraph{Induction step.}
 Let $k \in \Set{n+1,\ldots,n+m-1}$.  We determine
 $\maxFromToFn{{\N}_k}$ using IH. Let ${\N}_k$ be obtained from
 $\Set{{\N}_1,\ldots, {\N}_{k-1}}$ by splicing the two halves, $e_1$
 and $e_2$, of some original edge $e\in\edgesSharp{G}$.  We consider
 the two cases identified earlier in this section separately.
 
 \textbf{Case 1:}
   With no loss of generality, suppose $e_1\in\edgesIn{G_i}$
   and $e_2\in\edgesOut{G_j}$. Hence:
   \begin{alignat*}{5}
   & \edgesIn{G_k} &&=
       \ \ && \big(\edgesIn{G_i} - \Set{e_1}\big) \cup \edgesIn{G_j} ,
   \\[1.32ex]
   & \edgesOut{G_k} &&=
       \ \ && \edgesOut{G_i} \cup \big(\edgesOut{G_j} - \Set{e_2}\big) ,
   \\[1.32ex]
   & \edgesSharp{G_k} &&=
       \ \ && \edgesSharp{G_i} \cup \edgesSharp{G_j} \cup \Set{e}. 
   \end{alignat*}
 Consider arbitrary $A\subseteq\edgesIn{G_k}$ and $B\subseteq\edgesOut{G_k}$
 and let:
   \begin{alignat*}{5} 
   & A = A_1\cup A_2 \quad &&\text{and}\quad
     && B = B_1\cup B_2, \quad\text{where}
   \\[.91ex]
   & A_1\subseteq\edgesIn{G_i} - \Set{e_1}\quad &&\text{and}\quad
     && B_1\subseteq\edgesOut{G_i},
   \\[.91ex]   
   & A_2\subseteq\edgesIn{G_j} \quad &&\text{and}\quad
     && B_2\subseteq\edgesOut{G_j} - \Set{e_2} . 
   \end{alignat*}
 We then define:
   \begin{alignat*}{5} 
   & \maxFromTo{{\N}_k}{A,B}\ \triangleq\ \ 
   \\[1.01ex]   
   &\qquad \maxFromTo{{\N}_i}{A_1,B_1} + \maxFromTo{{\N}_j}{A_2,B_2}\ +
   \\[1.01ex]
   &\qquad \min\,\SET{\;
    \maxFromToAft{{\N}_i}{e_1,B_1}{A_1,B_1},
    \ \maxFromToAft{{\N}_j}{A_2,e_2}{A_2,B_2}\;}
   \end{alignat*}
The first line after `$\triangleq$' is the part of the maximum flow from $A$
to $B$ that does not use the edge $e$; the second line after `$\triangleq$'
is the part of the maximum flow from $A$ to $B$ that does use the edge $e$.
We have thus defined $\maxFromToFn{{\N}_k}$ in terms of the already-defined, by
IH, the functions $\maxFromToFn{{\N}_i}$ and $\maxFromToFn{{\N}_j}$, also invoking
Lemma~\ref{lem:relating-the-two-functions} which gives us 
$\maxFromToAftFn{}$ in terms of $\maxFromToFn{}$. The value of
$\maxFromTo{{\N}_k}{A,B}$ is obtained by using twice `$+$', once `$\min$',
and twice `$-$' for the invocations of $\maxFromToAft{{\N}_i}{e_1,B_1}{A_1,B_1}$
and $\maxFromToAft{{\N}_j}{A_2,e_2}{A_2,B_2}$
(see Lemma~\ref{lem:relating-the-two-functions}).

\textbf{Case 2:}
   With $e_1\in\edgesIn{G_i}$ and $e_2\in\edgesOut{G_i}$, we have in this case:
   \begin{alignat*}{5}
   & \edgesIn{G_k} &&=
       \ \ && \edgesIn{G_i} - \Set{e_1} ,
   \\[1.32ex]
   & \edgesOut{G_k} &&=
       \ \ && \edgesOut{G_i} - \Set{e_2} ,
   \\[1.32ex]
   & \edgesSharp{G_k} &&=
       \ \ && \edgesSharp{G_i} \cup \Set{e}. 
   \end{alignat*}   
Consider arbitrary $A\subseteq\edgesIn{G_k}$ and $B\subseteq\edgesOut{G_k}$.
Since $\edgesIn{G_k}\subseteq\edgesIn{G_i}$  and
$\edgesOut{G_k}\subseteq\edgesOut{G_i}$, we also have
$A\subseteq\edgesIn{G_i}$ and $B\subseteq\edgesOut{G_i}$. We then define:
   \begin{alignat*}{5} 
   & \maxFromTo{{\N}_k}{A,B}\ \triangleq\ \ 
   \\[1.01ex]   
   &\qquad \maxFromTo{{\N}_i}{A,B}\ +
   \\[1.01ex]
   &\qquad \min\,\SET{\;
    \maxFromToAft{{\N}_i}{A\cup\Set{e_1},B}{A,B},
    \ \maxFromToAft{{\N}_i}{A,B\cup\Set{e_2}}{A,B}\;}
   \end{alignat*}
The first line after `$\triangleq$' is the part of the maximum flow from $A$
to $B$ that does not use the edge $e$; the second line after `$\triangleq$'
is the part of the maximum flow from $A$ to $B$ that does use the edge $e$.
We have again defined $\maxFromToFn{{\N}_k}$ in terms of the already-defined, by
IH, functions $\maxFromToFn{{\N}_i}$ and $\maxFromToFn{{\N}_i}$, and again
invoking Lemma~\ref{lem:relating-the-two-functions} which gives us 
$\maxFromToAftFn{}$ in terms of $\maxFromToFn{}$.
The value of
$\maxFromTo{{\N}_k}{A,B}$ is obtained by using `$+$' once, `$\min$' once,
and twice `$-$' for the invocations of
$\maxFromToAft{{\N}_i}{A\cup\Set{e_1}}{A,B}$
and $\maxFromToAft{{\N}_i}{A,B\cup\Set{e_2}}{A,B}$
(see Lemma~\ref{lem:relating-the-two-functions}).

This completes the induction step and the definition of the function
$\maxFromToFn{{\N}_i}$ for every $i = 1,\ldots,n+m$.

The next lemma is used in the proof of Theorem~\ref{thm:our-result}.

\begin{lemma}
\label{lem:reassembling}
Consider the reassembling of the flow network
$\N = (G,c)$ described in the opening paragraph of
Section~\ref{sect:reassembling}. 
Let $\delta$ be the least upper bound of the resulting edge boundary
degrees $\Set{p_i+q_i\,|\,1\leqslant i\leqslant n+m}$. Then the
function $\maxFromToFn{{\N}} = \maxFromToFn{{\N}_{n+m}}$ is computed
in time $\bigO{(m+n)\cdot 2^{\delta}}$ using only three arithmetic
operations $\Set{\min,+,- }$.
\end{lemma}

\begin{proof}
For each $i = 1,\ldots,n+m$, the number of arguments
$(A,B) \in \edgesIn{G_i}\times\edgesOut{G_i}$ at which the
function $\maxFromToFn{{\N}_{i}}$ has to be determined is
$2^{p_i}\cdot 2^{q_i} = 2^{p_i+q_i} \leqslant 2^{\delta}$. And each
such determination is carried out using at most four times an
operation in $\Set{+,-}$ and at most twice an operation in
$\Set{\max,\min}$. A subtraction with `$-$' is involved with each invocation
of the function $\maxFromToAftFn{}$
(Lemma~\ref{lem:relating-the-two-functions}). Both `$+$' and `$\min$'
are involved in the determination of
$\maxFromToFn{{\N}_{1}}, \ldots, \maxFromToFn{{\N}_{n}}$,
and both `$-$' and `$\min$' are involved in the determination of
$\maxFromToFn{{\N}_{n+1}}, \ldots, \maxFromToFn{{\N}_{n+m}}$.
\end{proof}


\section{The Main Result}
\label{sect:our-result}

In order to use the algorithm whose existence is asserted in
Lemma~\ref{lem:reassembling-of-3-regular} in our main result
(Theorem~\ref{thm:our-result}), we need to transform the underlying
graph $G$ of the network $\N = (G,c)$ according to
Lemma~\ref{lem:equivalent-transformation}. Part 5 of the latter uses
the notion of \emph{edge-outerplanarity}, which we next define and
compare to the standard notion of \emph{outerplanarity}.

  We make a distinction between \emph{planar} graphs and \emph{plane}
  graphs.  $G$ is a \emph{plane graph} if it is drawn on the plane
  without any edge crossings. $G$ is a \emph{planar graph} if it is
  isomorphic to a plane graph; \ie, it is embeddable in the plane in
  such a way that its edges intersect only at their endpoints.  To keep
  the distinction between the two notions, we define
  the \emph{outerplanarity index} of a \emph{planar} graph and
  the \emph{outerplanarity} of a \emph{plane} graph.

  If $G$ is a plane graph, directed or undirected, then
  the \emph{outerplanarity} of $G$ is the number $k$ of times that all
  the vertices on the outer face (together with all their incident
  edges) have to be removed in order to obtain the empty graph. In
  such a case, we say that the plane graph $G$
  is \emph{$k$-outerplanar}.

  If $G$ is a planar graph, directed or undirected, then
  the \emph{outerplanarity index} of $G$ is the minimum of the
  outerplanarities of all the plane embeddings $G'$ of $G$.

  Deciding whether an arbitrary graph is planar can be carried out in
  linear time $\bigOO{n}$ and, if it is planar, a plane embedding of it
  can also be carried out in linear time~\cite{patrignani2013}.  Given a
  planar graph $G$, the outerplanarity index $k$ of $G$ and a
  $k$-outerplanar embedding of $G$ in the plane can be computed in time
  $\bigOO{n^2}$, and a $4$-approximation of its outerplanarity index can
  be computed in linear time~\cite{Kammer2007}.

\begin{definition}{Edge-Outerplanarity}
\label{def:edge-outerplanarity}
  Let $G$ be a plane graph, directed or undirected.
  If $\edges{G} =\varnothing$ and $G$ is a graph of isolated vertices,
  the \emph{edge outerplanarity} of $G$ is $0$. If $\edges{G} \neq\varnothing$,
  we pose $G_0 := G$ and define $K_0$ as the set of edges lying on $\OutF{G_0}$.

  For every $i>0$, we define $G_i$ as the plane graph obtained
  after deleting all the edges in $K_0 \cup \cdots \cup K_{i-1}$ from the
  initial $G$ and $K_i$ the set of edges lying on $\OutF{G_i}$.
  
  The \emph{edge outerplanarity} of $G$, denoted $\OutPlan{E}{G}$, 
  is the least integer $k$ such that $G_{k}$ is a graph without edges,
  \ie, the edge outerplanarity of $G_{k}$ is $0$. This process of peeling
  off the edges lying on the outer face $k$ times produces a $k$-block
  partition of $\edges{G}$, namely, $\Set{K_0,\ldots,K_{k-1}}$.%
   \footnote{There is an unessential difference between our definition here and
     the definition in~\cite{bentz2009}. In Section 2.2 of that
     reference, ``a $k$-edge-outerplanar graph is a planar graph
     having an embedding with \emph{at most} $k$ layers of edges.'' In
     our presentation, we limit the definition to plane graphs and say ``a
     $k$-edge-outerplanar plane graph has \emph{exactly} $k$ layers of
     edges.'' Our version simplifies a few things later.}
\end{definition}

  To keep \emph{outerplanarity} and \emph{edge outerplanarity} clearly
  apart, we call the first \emph{vertex outerplanarity}, or more
  simply \emph{V-outerplanarity}, and the second \emph{edge
  outerplanarity}, or more simply \emph{E-outerplanarity}.

  There is a close relationship between \emph{V-outerplanarity} and
  \emph{E-outerplanarity} (Theorem 4 in Section 5.1
  in~\cite{bentz2009}).  In the case of three-regular plane graphs,
  the relationship is much easier to state. This is
  Proposition~\ref{prop:V-outer-vs-E-outer} next, not needed for our
  main result (Theorem~\ref{thm:our-result}) but included here for
  completeness.

  \begin{proposition}
  \label{prop:V-outer-vs-E-outer}
  If $G$ is a $3$-regular plane graph, directed or undirected, then:
  \[
     \OutPlan{V\!\!}{G} \leqslant \OutPlan{E}{G} 
     \leqslant 1+ \OutPlan{V\!\!}{G}.
  \]
  Thus, for $3$-regular plane graphs,
  \emph{V-outerplanarity} and \emph{E-outerplanarity} are ``almost the same''.
\end{proposition}

\begin{sketch}
  For a $3$-regular plane graph, the difference between
  $\OutPlan{V\!\!}{G}$ and $\OutPlan{E}{G}$ occurs in the last stage
  in the process of repeatedly removing (in the case of standard
  \emph{V-outerplanarity}) all vertices on the outer face and all their incident
  edges. The corresponding last stage in the case
  of \emph{E-outerplanarity} may or may not delete all edges; if
  it does not, then one extra stage is needed to delete all remaining
  edges.
\end{sketch}

\begin{lemma}
\label{lem:equivalent-transformation}
There is an algorithm which, given an arbitrary flow
network $\N = (G,c) $, returns a flow network
$\transA{\N} = (\transA{G},\transA{c})$ in time $\bigOO{n+m}$,
where $n = \size{\vertices{G}}$ and $m = \size{\edgesSharp{G}}$, such that:
\begin{enumerate}[itemsep=1pt,parsep=2pt,topsep=2pt,partopsep=0pt]
\item $\edgesIn{\transA{G}} = \edgesIn{G}$
      and $\edgesOut{\transA{G}} = \edgesOut{G}$, so that
      also $\edgesIO{\transA{G}} = \edgesIO{G}$.
\item $\size{\vertices{\transA{G}}} = \bigOO{n}$ and
      $\size{\edgesSharp{\transA{G}}} = \bigOO{m}$ .
\item $\N$ and $\transA{\N}$ are equivalent flow networks, in particular, \\
      $\tau: \Power{\edgesIO{G}} \to\intervals{\reals}$ is a
      principal typing for $\N$ iff
      it is a principal typing for $\transA{\N}$.
\item $\transA{G}$ is a $3$-regular directed graph without two-edge cycles.%
      \footnote{See footnote~\ref{foot:two-edge-cycle} on
      page~\pageref{foot:two-edge-cycle}.}
\end{enumerate}
Moreover, if $G$ is a plane graph, then:
\begin{itemize}[itemsep=1pt,parsep=2pt,topsep=2pt,partopsep=0pt]
\item[5.] $\transA{G}$ is a plane graph such that
          $\OutPlan{E}{\transA{G}} = \OutPlan{E}{G}$.
\end{itemize}
\end{lemma}

It is worth pointing out that the hidden constants in the big-O notations
above are small integers, each a single-digit number.

\begin{proof}
  This is shown in Section 3 of the earlier report~\cite{kfoury2018b}. The
  $5$-part conclusion of the lemma here is divided into several lemmas in
  the earlier report.
\end{proof}

  Let $G$ be a simple undirected graph. A \emph{reassembling} of $G$
  is a rooted binary tree $\B$ whose nodes are subsets of
  $\vertices{G}$ and whose leaf nodes are singleton sets, with each of
  the latter containing a distinct vertex of $G$.  The parent of two
  nodes in $\B$ is the union of the two children's vertex sets. The
  root node of $\B$ is the full set $\vertices{G}$. If
  $n = \size{\vertices{G}}$, there are thus $n$ leaf nodes in $\B$ and a
  total of $(2n -1)$ nodes in $\B$. We denote the reassembling of $G$
  according to $\B$ by writing $(G,\B)$.%
  \footnote{To keep apart $\B$ and $G$, we reserve the words `node'
  and `branch' for the tree $\B$, and the words `vertex' and `edge'
  for the graph $G$.}

  The \emph{edge-boundary degree} of a node in $\B$ is the number of
  edges that connect vertices in the node's set to vertices not in the
  node's set. Following a terminology used in earlier reports,
  the \emph{$\alpha$-measure} of the reassembling $(G,\B)$, denoted
  $\alpha(G,\B)$, is the largest edge-boundary degree of any node in
  the tree $\B$. We say $\alpha(G,\B)$ is \emph{optimal}
  if it is minimum among all $\alpha$-measures of $G$'s reassemblings,
  in which case we also say $\B$ is \emph{$\alpha$-optimal}.%
  \footnote{The reassembling process described in the Introduction,
  Section~\ref{sect:intro}, and again in the opening paragraph of
  Section~\ref{sect:reassembling}, is a \emph{lazy} version of the
  reassembling defined here. We can call the latter the \emph{eager}
  version of reassembling. The difference is that, in the lazy
  version, only one edge's two halves are spliced at any given time;
  in the eager version defined in this section, the two halves of all
  the edges between two disjoint components (\ie, two sibling nodes in
  the tree $\B$) are spliced simultaneously. Hence, if we carry out
  the reassembling $(G,\B)$ lazily, then a least upper bound on the
  edge-boundary degrees of all the components is $2\cdot\alpha(G,\B) -
  1$. }

  The problem of constructing an $\alpha$-optimal reassembling
  $(G,\B)$ of a simple undirected graph $G$ in general was already
  shown NP-hard~\cite[among
  others]{kfoury+mirzaei:2017,kfoury+mirzaei:2017B}.  However,
  restricting attention to \emph{plane} graphs, we have the following
  positive result.

\begin{lemma}
\label{lem:reassembling-of-3-regular}
There is an algorithm which, given a plane $3$-regular simple undirected
graph $G$, returns a reassembling $(G,\B)$ in time $\bigOO{n}$
such that $\alpha(G,\B) \leqslant 2k$,
where $k = \OutPlan{E}{G}$ and $n = \size{\vertices{G}}$.
\end{lemma}

\begin{proof}
This is Theorem 9 and Corollary 20 in the report~\cite{kfoury+sisson2018}.
\end{proof}

\begin{theorem}
  \label{thm:our-result}
  There is an algorithm which, given a flow network $\N = (G,c)$ where
  $G$ is planar, computes a principal typing for
  $\N$ in time $\bigOO{n\cdot 2^{\delta}}$, where $n = \size{\vertices{G}}$,
  $\delta = \max\,\SET{\,2k ,\, \size{\edgesIn{G}\cup\edgesOut{G}}\,}$
  and $k = \OutPlan{E}{H}$ where $H$ is a plane embedding of $G$.
\end{theorem}

\begin{proof}
  We start by computing a plane embedding $H$ of $G$, which can be
  done in time $\bigOO{n}$, as pointed out at the beginning of this
  section. After this embedding, we refer to the network $(H,c)$ by
  the same name `$\N$'. Next, we use
  Lemma~\ref{lem:equivalent-transformation} to transform the network
  $\N = (H,c)$ into an equivalent network $\transA{\N} =
  (\transA{H},\transA{c})$ where $\transA{H}$ is a $3$-regular plane
  graph such that $k = \OutPlan{E}{H} = \OutPlan{E}{\transA{H}}$.  The
  transformation $\N \mapsto \transA{\N}$ is carried out in time
  $\bigOO{n+m}$ and therefore in time $\bigOO{n}$, because $H$ is a
  plane graph.

  Next, we compute a reassembling $(\transA{H},\B)$ in time $\bigOO{n}$,
  by invoking Lemma~\ref{lem:reassembling-of-3-regular}, with
  $\alpha(\transA{H},\B)\leqslant 2k$. We now use
  Lemma~\ref{lem:reassembling} to compute the function
  $\maxFromToFn{\transA{\N}} = \maxFromToFn{{\N}}$ in time
  $\bigO{(m+n)\cdot 2^{\delta}}$ and therefore in time $\bigO{n\cdot 2^{\delta}}$
  where $\delta =
  \max\,\SET{\,2k ,\, \size{\edgesIn{\transA{H}}\cup\edgesOut{\transA{H}}}\,} =
  \max\,\SET{\,2k ,\, \size{\edgesIn{G}\cup\edgesOut{G}}\,}$.

  Finally, we use Lemma~\ref{lem:typing} to return a principal typing $\tau$
  for $\N$, simultaneously with the computation of the function
  $\maxFromToFn{{\N}}$.
\end{proof}

  It is worth pointing out that the computation of principal typings in
  Theorem~\ref{thm:our-result} involves only three arithmetic operations
  $\Set{\,\min,\,+,\,-\,}$, according to Lemma~\ref{lem:reassembling}.


\section{Future Work}
\label{sect:future}

Flow networks in this report are the simplest possible and are of the
form $\N = (G,c)$, where the function $c : \edges{G}\to\nreals$
assigns an upper-bound capacity to every edge. The method proposed in
this report to compute principal typings for such networks, in
fixed-parameter linear time, is a `template' for further
extensions to more general forms of flow networks.

The next extension of the method considers flow networks of the form
$\N = (G, \lowerB, \upperB)$, where the two functions
$\lowerB, \upperB : \edges{G}\to\nreals$ assign a lower-bound capacity
and an upper-bound capacity, respectively, to every edge. And there
are still other extensions under consideration, including the
following:
\begin{itemize}[itemsep=1pt,parsep=2pt,topsep=2pt,partopsep=0pt]  
\item
  \emph{multicommodity flows}
  (formal definitions in~\cite[Chapt. 17]{1993orlin}),
\item
  \emph{minimum-cost flows}, \emph{minimum-cost max flows}, and variations
  (definitions in~\cite[Chapt. 9-11]{1993orlin}),
\item  
  \emph{flows with multiplicative gains and losses}, also called
  \emph{generalized flows} (definitions in~\cite[Chapt. 15]{1993orlin}),
\item  
  \emph{flows with additive gains and losses}
  (definitions in~\cite{brandenburg2011shortest}).  
\end{itemize}
This is on-going work requiring various refinements, not all the same for
the different extensions.


\ifTR
\else
\fi


{\footnotesize 
\addcontentsline{toc}{section}{References} 
\bibliographystyle{plain} 
\bibliography{./fixed-parameter-linear-time-for-typings}
}

\ifTR
\else
\fi

\Hide{

\newpage  
\appendix

\section{Appendix: Pseudocode of Algorithm $\KS$}
  \label{appendix:pseudocode}
  }

\end{document}